\documentclass[aps,pra,twocolumn,showkeys]{revtex4-1}
\usepackage{amsmath,paralist,amsthm,comment,amssymb}
\usepackage{times,color}

\DeclareMathOperator{\Tr}{Tr}

\newcommand{\mc}[1]{\mathcal{#1}}
\newcommand{\mbf}[1]{\mathbf{#1}}
\newcommand{\mbb}[1]{\mathbb{#1}}

\newcommand{\F}{\mathbb{F}}
\newcommand{\nix}[1]{}
\newcommand{\ket}[1]{|#1\rangle}

\newcommand{\be}{\begin{eqnarray*}}
\newcommand{\ee}{\end{eqnarray*}}
\newcommand{\ben}{\begin{eqnarray}}
\newcommand{\een}{\end{eqnarray}}

\newcommand{\ba}{\begin{array}}
\newcommand{\ea}{\end{array}}

\newtheorem{theorem}{Theorem}
\newtheorem{corollary}[theorem]{Corollary}
\newtheorem{lemma}[theorem]{Lemma}

\newtheorem{fact}{Fact}
\newtheorem*{remark}{Remark}

\newtheorem*{defn}{Definition}

\newcommand{\bra}[1]{\langle #1|}

\begin{document}


\title{ Bounds on the Information Rate of Quantum Secret Sharing Schemes}


\author{Pradeep Sarvepalli}
\email[]{pradeep@phas.ubc.ca}
\affiliation{Department of Physics and Astronomy, University of British Columbia, Vancouver, BC V6T 1Z1}

\date{July 7, 2010}

\begin{abstract}
An important metric of the performance of a quantum secret sharing scheme is its 
information rate. Beyond the fact that the information rate is upper bounded by one, 
very little is known in terms of bounds on the  information rate
of quantum secret sharing schemes. Further, not every scheme can be realized with
rate one. In this paper we derive new upper bounds for the information rates of quantum secret
sharing schemes. We show  that there exist quantum access structures on $n$ players for which 
the information rate cannot be better than  $O((\log_2 n)/n)$. 
These results are the quantum analogues of  the bounds for classical secret sharing schemes proved by Csirmaz.
\end{abstract}


\maketitle
\section{Introduction}

Quantum secret sharing is a cryptographic protocol to distribute a secret state (either classical or quantum) among a group of 
players $P$, such that only authorized subsets of $P$ can reconstruct the secret state  from the distributed quantum states \cite{hillery99,cleve99}. 
The quantum state distributed to each
party is called a share.  The collection of authorized sets is called the access structure of the scheme.
Two important problems related to quantum secret sharing are the construction of efficient secret
sharing schemes for various access structures \cite{gottesman00,smith00} and establishing the bounds on the efficiency of secret sharing schemes \cite{gottesman00,imai04}. 
Despite the rapid growth of the field since its inception in \cite{hillery99}, see for instance  \cite{cleve99, gottesman00, smith00, karlsson99, imai04, lance04, rietjens05, markham08, keet10,benor06, gaertner02, tittel01} and the references therein,  bounds on the efficiency of quantum secret sharing schemes are hard to come by.
The main purpose of this paper is to report some progress on this problem.

The efficiency of a quantum secret sharing scheme is quantified by its information rate. Informally, this is defined as the ratio of the size of the secret to the size of the largest share, (a precise definition will be given later). 
The importance of the information rate can be understood as follows. Smaller the rate, larger are the sizes of the shares and the overhead costs of storage and communication. As the shares are to be kept secret, the security of the protocol can be undermined by large shares. For these reasons it is beneficial to design schemes with high information rate. 

A secret sharing scheme is said to be perfect if unauthorized sets cannot extract any information about the secret. In such schemes the size of the share must be at least the
size of the secret \cite{gottesman00,imai04}. Therefore the information rate of perfect quantum secret sharing  schemes is upper bounded by one. 
Given an access structure it  is not always possible to construct schemes which realize
the access structure with information rate one. Therefore, we would like to know bounds 
on the sizes of the shares for a given access structure. Both Gottesman \cite{gottesman00} 
and Smith \cite{smith00} have given constructions for general access structures. 
These constructions provide implicitly lower bounds for the information rate of quantum 
secret sharing schemes. In these constructions, the size of a share can be exponentially 
larger than the size of the secret.

At the other end we could ask what is the minimum size of share ``required'' to implement an access structure on $n$
participants. This question is difficult to answer unless we impose some restrictions on the access structures. 
We study the information rate of a specific type of access structures
and show that for these the size of a share must be at least $O(n/\log_2 n)$ larger than the size of the secret. 
These results are quantum analogues of the bounds for classical secret sharing schemes proved by Csirmaz \cite{csirmaz97}. 

In this paper we restrict our attention to perfect quantum secret sharing schemes where the secret is a quantum state. 
Bounds on quantum secret sharing schemes which share a classical secret are not dealt with
in this paper. 

\subsection{Background}

In this section we will briefly review some of the necessary background and notation used in the paper. 
We denote the participants of a  secret sharing scheme by $P$, typically assumed to be the set $P=\{1,2,\ldots,n \}$. 
The access structure of a secret sharing
is denoted by $\Gamma$. A special type of authorized set is one whose proper subsets are
unauthorized. The collection of minimal authorized sets of an access structure $\Gamma$
is denoted as $\Gamma_{\min}$. 
The notation $2^P$ is for the collection of all subsets of $P$.
We denote the complement of $A\subseteq P $ by $\bar{A}=P\setminus A$.
An access structure is said to be monotone if and only if for any $A\in \Gamma$, any set containing $A$
is also in $\Gamma$. 
Given $\Gamma$ we define the dual access structure as $\Gamma^\ast= \{ \bar{A} \mid A\not\in \Gamma\}$.
An access structure is said to be self-dual if $\Gamma=\Gamma^\ast$.

An access structure that can be realized by quantum secret sharing schemes is 
called a quantum access structure. 
Due to the no-cloning theorem \cite{wootters82,dieks82}, quantum access structures satisfy the following requirements:

\begin{fact}[\cite{cleve99}]\label{fc:authComp}
In a quantum access structure, the complement of an authorized set is unauthorized.
\end{fact}

\begin{fact}[\cite{gottesman00}]\label{fc:qAcc}
An access structure is a quantum access structure if and only if  it is monotone and no two authorized sets are disjoint.
\end{fact}

Our analysis of quantum secret sharing schemes is information theoretic. Therefore, we quickly recall some of the relevant notions 
of quantum  information. Given a density matrix $\rho$ we denote its von Neumann entropy as $S(\rho)$ which is defined as
\ben
S(\rho)&=& -\Tr(\rho \log_2 \rho).\label{eq:entropy1}
\een
Sometimes we have  to find the entropy of a set of quantum systems 
indexed by a set $A$. In such situations, we denote the
von Neumann entropy of a system $A$ as $S(A)$ where 
\ben
S(A) = -\Tr(\rho_A \log_2 \rho_A), \label{eq:entropy2}
\een
and  $\rho_A$ is the (reduced) density matrix of $A$.
In this paper unless otherwise specified entropy refers to the von Neumann entropy.
The mutual information between two systems $A$ and $B$ is denoted by $I(A:B)$
and
\ben
I(A:B) = S(A)+S(B)-S(A,B).\label{eq:mutual}
\een
A very useful inequality relating the entropy of three systems $A,B,C$ is the strong subadditivity inequality:
\ben
S(A,C)+S(B,C)\geq S(A,B,C)+S(C).\label{eq:strSubAdd}
\een
A slightly different statement of the above inequality will be useful to us. Because we
can associate the arguments of $S(\cdot)$ to an index set $V$, we can also regard 
entropy as a set function i.e.  $S:2^V\rightarrow \mbb{R}$. Then 
it becomes clear  that  entropy is a submodular function, in other words, it satisfies
\ben
S(X)+S(Y) \geq S(X\cup Y) +S(X\cap Y).\label{eq:subModularity}
\een
To show this let $C=X\cap Y$, $A=X\setminus C$ and $B=Y\setminus C$.
Then substituting for $A,B,C$ in equation~\eqref{eq:strSubAdd} we obtain
$S(X\setminus (X\cap Y), X\cap Y)+S(Y\setminus (X\cap Y),X\cap Y) \geq S(X\setminus (X\cap Y), Y\setminus (X\cap Y), X\cap Y)+S(X\cap Y)$ which simplifies to equation~\eqref{eq:subModularity}.

An information theoretic model for quantum secret sharing schemes was proposed in 
\cite{imai04}. As we make significant use of this  framework to prove our results
we provide a brief review of this model. 

We denote the Hilbert space of a system $A$ by $\mc{H}_A$. Let $S$ be the secret to
be distributed and let $\mc{H}_S$ be the associated Hilbert space. We assume that 
$\rho_S= \sum_{i\in \F_q} \alpha_i\ket{i}\bra{i}$.  The
secret is purified using a reference system $R$ to give a pure state $\ket{RS}$. Then 
\ben
S(R)=S(S)\label{eq:entropyRef}
\een
Assume that the set of players is given by $P=\{1,\ldots, n \}$.
A distribution of shares for the quantum secret sharing scheme is defined to be a completely positive trace preserving map $\Delta$
\ben
\Delta:\mc{H}_{RS}\rightarrow \mc{H}_R \otimes \mc{H}_P,
\een
such that for every authorized set $A\subseteq P$ there exists a recovery map $\mbf{1}_R\otimes \mc{R}_A:\mc{H}_R\otimes \mc{H}_{A}\rightarrow \mc{H}_{RS}$ that 
maps  $\rho_{RA} \mapsto \ket{RS}\bra{RS}$.
The information rate of the quantum secret sharing scheme is defined as 
\ben
\kappa = \frac{S(S)}{\max_i S(i)}.\label{eq:infoRate}
\een

The information theoretic model provides the following 
characterization of quantum secret sharing schemes \cite{imai04}: 

\begin{lemma}\label{lm:qssCond}
In a perfect quantum secret sharing scheme with access structure $\Gamma$ we must have
\ben
I(A:R)&=&I(A:S) = 2S(S)\mbox{ for all } A\in \Gamma \label{eq:qssCond1}\\
I(A:R)&=&I(A:S)=0 \mbox{ for all } A\not\in \Gamma \label{eq:qssCond2}
\een
\end{lemma}
The first condition has been called the recoverability requirement while the second has
been termed secrecy requirement. We also need the following lemma, although an immediate consequence 
of Lemma~\ref{lm:qssCond},  we include it for completeness.
\begin{lemma}\label{lm:abarReln}
Let  $A\in \Gamma$. Then we have the following relations.
\ben
S(A)&=& \frac{1}{2}I(A:\bar{A})+S(S)\\
S(\bar{A})&=& \frac{1}{2}I(A:\bar{A})\\
S(A)-S(\bar{A})&=&S(S)
\een
\end{lemma}
\begin{proof}
If $A\in \Gamma$, then $\bar{A}\not\in \Gamma$. Computing $I(A:\bar{A})$ we obtain
\ben
I(A:\bar{A}) &=& S(A)+S(\bar{A})-S(A,\bar{A})\label{eq:mut1}\\
&=& S(A)+S(\bar{A})-S(R)\\\nonumber
&=& S(A)+S(A, R ) -S(S)\\\nonumber
&=& S(A)+S(A)+S(R)-I(A:R)-S(S)\\ \nonumber
&=& 2S(A)-2S(S),\nonumber
\een
which gives us the first relation. Substituting for $S(A)$ in equation~\eqref{eq:mut1}
gives us the second relation. 
\be
I(A:\bar{A}) = I(A:\bar{A})/2+S(S) +S(\bar{A}) -S(S)
\ee
The last relation is immediate from the previous two equations.
\end{proof}

\section{Bounds on the Information Rate}
In this section we derive lower bounds on the size of shares (in terms of the von Neumann entropy) and thus compute upper bounds on the information rate. The main result of this
section is a lower bound on the size of shares for a class of access structures.
A similar result was shown for classical secret sharing schemes by Csirmaz \cite{csirmaz97}. 

\begin{lemma}\label{lm:jointEntropy}
Given a perfect quantum secret sharing scheme with access structure $\Gamma$,
for any  $A\subseteq P$, one of the following must hold:
\ben
S(A,R)&=& S(A)-S(S) \mbox{ if } A\in \Gamma\\
S(A,R)&=& S(A)+S(S) \mbox{ if } A\not\in \Gamma
\een
\end{lemma}
\begin{proof}
For any subset $A$, using the definition of mutual information in equation~\eqref{eq:mutual}
\ben
S(A,R)&=& S(A)+S(R)-I(A:R)\nonumber\\
S(A,R)&=& S(A)+S(S)-I(A:R),\label{eq:jointEntropy}
\een
where we have used equation~\eqref{eq:entropyRef} in the last step.
By Lemma~\ref{lm:qssCond}, we know that $I(A:R)=2S(S)$ if $A\in \Gamma$ and $I(A:R)=0$
otherwise. Substituting these values in equation~\eqref{eq:jointEntropy} we get the stated result. 
\end{proof}

\begin{lemma}\label{lm:subAddAuthSets}
Let $A, B $ be any two authorized sets such that $A\cap B $ is not authorized. Then we have the following inequality:
\ben
S(A)+S(B)\geq S(A\cup B) +S(A\cap B)+2S(S).\label{eq:subAddAuthSets}
\een
\end{lemma}
\begin{proof}
By using Lemma~\ref{lm:jointEntropy}, we can write $S(A)+S(B)$ as
\be
S(A)+S(B)&=& S(A\cup R)+S(B\cup R) + 2S(S)\\
&\geq &S(A\cup B \cup R)+S((A\cap B)\cup R)+2S(S)
\ee
where we used the subadditivity inequality in the last equation.
Once again using Lemma~\ref{lm:jointEntropy} and simplifying, we obtain
\be
S(A)+S(B)&\geq &S(A\cup B)
+S((A\cap B))+2S(S)
\ee
as stated. 
\end{proof}

\subsection{The General Case with $n$ Parties}
\begin{defn}[Csirmaz Access Structure]
Given an integer $n\geq 4$,  let $k$ be the largest integer such that $n\geq 2^k-2+k$.
Let $A\subset P $ be such that $|A|=k$. Consider all the $2^k$ subsets of $A$ enumerated in the order of decreasing cardinality. Let $B=P\setminus A=\{b_1,\ldots, b_{n-k} \}$. Let
$B_0=\emptyset $, $B_{2^k-2}=B$  and $B_i=\{b_1,\ldots,b_i \}$ for $0<i<2^{k}-2$.
Let $\Gamma_{\min}^{(n)}$ be the minimal access structure given as below:
\ben
\Gamma_{\min}^{(n)} =\{ A_i\cup B_i \mid 0\leq i <2^k-1\}.\label{eq:csirmazAcc}
\een
\end{defn}
We emphasize that in these access structures, we have $A_i\cap B_j=\emptyset$, $A_i\not\subset A_j$ if $i<j$ and 
$B_i\not\subset B_j$ if $i>j$.
As an aside we note that the access structure defined above departs from the one originally proposed by Csirmaz \cite{csirmaz97}
in the definition of the minimal authorized set $B_{2^k-2}$ and secondly, it is always connected i.e., every 
party occurs in some minimal authorized set.
It is easy to see that this is a minimal access structure. If 
it is not, then for some  distinct $i,j$ we must have  $A_i\cup B_i  \subsetneq A_j \cup B_j$, then it follows that 
$A_i \subseteq A_j$ and $B_i\subseteq B_j$. This implies $i\geq j$ and $i\leq j$, which is impossible for $i\neq j$. 
\begin{lemma}
The access structure $\Gamma_{\min}^{(n)}$, defined in equation~\eqref{eq:csirmazAcc},  is a quantum access structure. 
\end{lemma}
\begin{proof}
The  elements of  the minimal access structure $\Gamma_{\min}^{(n)}$ 
are given by 
\ben
M_i =A_i \cup B_i \mbox{ for } 0\leq i<2^k-1
\een
Consider $M_i \cap M_j$ for any $i\neq j$. Then because $A_k \cap B_l =\emptyset $ 
for any $k,l$ we have
\be
M_i \cap M_j = (A_i \cap A_j) \cup (B_i \cap B_j).
\ee
Without loss of generality we can assume that $i<j$. If $i=0$, then 
$B_0=\emptyset$ and $A_0=A\supsetneq A_j$ for all $j$. Therefore 
$M_0\cap M_j = A_j\neq \emptyset$. If $i\neq 0$, then we have $B_i\neq \emptyset$
and by construction $B_i\subsetneq B_j$. Thus we have $M_i \cap M_j \supseteq B_i\neq \emptyset$. Thus we always have $M_i\cap M_j\neq \emptyset$. Thus by Fact~\ref{fc:qAcc}$, \Gamma_{\min}^{(n)}$ is a quantum access structure.
\end{proof}

It turns out that the access structure $\Gamma^{(n)}$ is not self-dual. 
For instance, the set $A_i\cup B_{i-1}$ for $i>1$ is unauthorized as well its complement 
$(A\setminus A_i) \cup (B\setminus B_{i-1})$ is also unauthorized. 
For a technical reason in the computation of the
entropies, it is much more convenient to work with self-dual access structures. For this reason we need the following notion 
introduced by Gottesman \cite{gottesman00}.

\paragraph*{Purification of Access Structures:}
An access structure $\Gamma$ on a set of players $P=\{1,2,\ldots, n \}$ that is not self-dual can be ``purified'' to give a self-dual access structure $\overline{\Gamma}$ on a set of  players $P'=\{ 1,2, \ldots,n, n+1\}$ by 
\begin{compactenum}[i)]
\item Adding an additional party, $n+1$.
\item Additional authorized sets, created from unauthorized sets of original structure and having the following
form:
$$
\{A\cup \{ n+1\}  \mid A\not\in \Gamma; P\setminus A \not\in \Gamma \}.
$$
\end{compactenum}

The additional authorized sets all contain the party $\{n+1\}$. These sets are created from a pair sets $A$ and $P\setminus A$ such that both are not in $\Gamma$, i.e., both are unauthorized. Purification converts one of them to an authorized set in $\overline{\Gamma}$ 
along with additional party. Thus $A\cup \{ n+1\} \in \overline{\Gamma}$ but $P \setminus A \not\in \overline{\Gamma}$. 
\begin{remark}
It is important to note that all the authorized sets of $\Gamma$ are also authorized sets of 
its purification $\overline{\Gamma}$ and all the unauthorized sets of $\Gamma$ are unauthorized in 
$\overline{\Gamma}$.
\end{remark}

Purification of access structures was introduced by Gottesman in \cite{gottesman00}. Purification of access structures is useful
in that it gives an access structure that has a  pure-state secret sharing scheme and the scheme for the original access structure
can be obtained by discarding the share that has been added. We recover the original access structure by discarding the share
associated with $\{ n+1\}$.

The following theorem closely follows the structure of 
\cite[Theorem~3.2,~Lemma~3.3--3.4]{csirmaz97}, but note that the results therein do not 
apply in the quantum setting. Firstly, because the model of quantum secret sharing 
schemes is different and secondly, despite the similarity of the access structures, 
the proofs in \cite{csirmaz97} invoke the use a (polymatroidal) function which is assumed 
to be  monotonic and submodular. The von Neumann entropy, as is well-known, does
not satisfy the assumption of monotonicity.

\begin{theorem}\label{th:shareSize}
Let $\overline{\Gamma}^{(n-1)}$ be a purification of an access structure defined as in 
equation~\eqref{eq:csirmazAcc}. Then in any realization of $\overline{\Gamma}^{(n-1)}$ there exists some share $i$ for which 
$S(i)/S(S) \geq ((2^{k+1}-1)/(2k+1))=O(n/\log_2 n)$.
\end{theorem}
\begin{proof}
We assume that $\Gamma^{(n-1)}$ is defined on $P=\{1,2,\ldots,n-1 \}$ and its purification $\overline{\Gamma}^{(n-1)}$ is defined on 
$P'=\{1,2,\ldots, n-1, n \}$
The idea behind the proof is to obtain a lower bound on the entropy of a subset of $P'$.
Our aim will be to get a lower bound on the entropy of a subset $A\subseteq \{ 1,\ldots, n-1,n\}$ such as $S(A)\geq \alpha S(S)$. Then using the fact that $S(A)\leq S(i) |A|$
for some $i\in P'$, we can lower bound the size of the $i$th share as $\alpha S(S)/|A|$.

Let $ X_i= B_{i}\cup A$ and $Y_i=  A_{i+1}\cup B_{i+1}$. First observe that $X_i, Y_i$ are both
authorized sets as $X_i\supset M_{i}$ and $Y_i\supset M_{i+1}$. Furthermore, $X_i\cap Y_i = A_{i+1}\cup B_i $ is unauthorized. Suppose on the contrary that 
$X_i\cap Y_i$ is authorized; then for some $j$ we must have $M_j=A_j\cup B_j \subseteq A_{i+1}\cup B_i$. 
Since $A_l\cap B_m=\emptyset$ for all choice of $l$ and $m$, this
implies that  $A_j \subseteq A_{i+1}$ and $B_j\subseteq B_i$, 
which can only hold if  $ j \geq i+1$ and $j\leq i $  giving us a contradiction.
Therefore $X_i\cap Y_i $ is unauthorized and we can apply  Lemma~\ref{lm:subAddAuthSets} which
gives us 
\be
S(X_i)+S(Y_i) &\geq &S(X_i\cup Y_i) + S(X_i\cap Y_i) + 2S(S).
\ee
This can be rewritten as
\be
S(A\cup B_i)+S(A_{i+1}\cup B_{i+1})&\geq& S(A\cup B_{i+1})+2S(S)\\
&+&S(A_{i+1}\cup B_i).
\ee
Now we apply submodularity inequality to the sets $ A_{i+1}\cup B_i$ and $B_{i+1}$:
\be
S( A_{i+1}\cup B_i)+S(B_{i+1}) \geq S( A_{i+1}\cup B_{i+1})+S(B_i).
\ee
Adding the previous two inequalities we obtain
\be
S(A\cup B_{i})+S(B_{i+1})\geq  S(A\cup B_{i+1})+S(B_i)+2S(S)
\ee
Clearly, this inequality holds for all $0\leq i<2^k-2$. Adding them all up we obtain
\be
\sum_{i=0}^{2^k-3}S(A\cup B_{i})+S(B_{i+1})&\geq& \sum_{i=0}^{2^k-3} S(A\cup B_{i+1})+S(B_i)\\
&+&(2^{k}-2)2S(S)\\
S(A\cup B_0)+S(B_{2^k-2}) &\geq& S(A\cup B_{2^k-2})+S(B_0)\\
&+&(2^k-2)2S(S)
\ee
As $B_0=\emptyset$, this reduces to 
\be
S(A)+S(B_{2^k-2}) &\geq& S(A\cup B_{2^k-2})+(2^k-2)2S(S)
\ee
Since $A\cup B_{2^k-2}$ is authorized its complement is unauthorized by Fact~\ref{fc:authComp} and by Lemma~\ref{lm:abarReln}, $S(A\cup B_{2^k-2}) = S(\overline{A\cup B_{2^k-2}})+S(S) \geq 2S(S)$,
where we also used the fact that  $S(Z)\geq S(S)$ for any unauthorized set $Z$, 
\cite[Theorem~6]{imai04}, see also \cite[Theorem~4]{gottesman00}. 
Therefore we now have 
\be
S(A)+S(B_{2^k-2}) \geq (2^k-1)2S(S)
\ee
By Lemma~\ref{lm:abarReln}, $S(B_{2^k-2})=S(\overline{B}_{2^k-2})=S(A,R, n)\leq S(A,R)+S(R,n)-S(S) = S(A)+S(n)-S(S)$,  therefore we obtain
\be
2S(A)+S(n)-S(S) &\geq& (2^k-1)2S(S)\\
2S(A)+S(n)&\geq& (2^{k+1}-1)S(S)
\ee
But $S(A)\leq \sum_{i=1}^{k} S(i)$, hence 
\be
S(n)+\sum_{i=1}^k 2S(i)  \geq ({2^{k+1}-1})S(S).
\ee
Hence for some $1\leq i\leq k$ or $i=n$, we must have 
\be
S(i) \geq \frac{2^{k+1}-1}{2k+1}S(S)
\ee
Since $\Gamma^{(n-1)}$ has $n-1$ participants with
$n-1\geq 2^k+k-2$, we see that some share is at least as large as 
$O(n/\log_2 n)$ the size of the secret. 
\end{proof}

\begin{corollary}
For all $n \geq 4$, there exist quantum secret sharing schemes with information rate upper bounded by $O((\log_2 n)/n)$.
\end{corollary}

Classically the techniques used to study of bounds on secret sharing schemes often rely on 
information theoretic inequalities of the Shannon entropy. Furthermore, there is a larger set of tools available
such as polymatroids that make it possible to derive general results. Unfortunately in the study of quantum secret sharing schemes, these tools are either difficult to apply or either inapplicable.  For this reason it is interesting 
that we have been able to use the von Neumann entropic inequalities to prove Theorem~\ref{th:shareSize}.

\subsection {Bounds for Special Cases:  $n=4,5$ }
Often the study of small instances can reveal interesting insights. For this reason
we now consider the access structure for $n=4$, to derive some slightly tighter bounds than the ones obtained in 
Theorem~\ref{th:shareSize}. 
 
\begin{theorem}
Any realization of the quantum minimal access structure $\Gamma^{(4)}$ 
or its purification must have some share of size $50\%$ larger than the size of secret.
\ben
\Gamma_{\min}^{4} =\{(1,2);(1,3); (2,3,4) \}\label{eq:gamma4}
\een 
\end{theorem}
\begin{proof}
The purification of $\Gamma_{\min}^{(4)}$ is given by 
$\overline{\Gamma}_{\min}^{(4)} =  \{ (1,2);(1,3); (2,3,4); (2,3,p); p(1,4,p) \}$.
First we observe that $\overline{\Gamma}_{\min}^{(4)}$ is  a realizable quantum access 
structure as any two authorized sets have a nonempty  intersection. 
Furthermore, it can be easily verified that it is also a self-dual access  structure. 
Consider now the sets $A_0=\{ 1,2,3\}$, $B_0=\{2,3,4\}$. Both these sets are authorized sets while $A_0\cap B_0 =\{ 2,3\}$ is not authorized.
Therefore by Lemma~\ref{lm:subAddAuthSets}, we have 
\be
S(1,2,3)+S(2,3,4)&\geq& S(1,2,3,4)+S(2,3)+2S(S).
\ee
Consider now the sets $A_1=\{2,3 \}$ and $B_1=\{3,4 \}$. By the subadditivity of the 
von Neumann entropy we obtain
\ben
S(2,3)+S(3,4)&\geq& S(2,3,4)+S(3).
\een
Adding the previous two inequalities we obtain
\ben
S(1,2,3)+S(3,4)&\geq& S(1,2,3,4)+S(3)+2S(S)
\een
This can be rewritten as 
\ben
S(1,2,3)-S(3)&\geq &S(1,2,3,4)-S(3,4)+2S(S) \label{eq:a0a1}
\een
Consider now the sets $A_2=\{1,2 \}$ and $B_2=\{1,3\}$ which are both in $\overline{\Gamma}_{\min}^{(4)}$, while $A_2\cap B_2=\{ 1\}$ is unauthorized.  Applying Lemma~\ref{lm:subAddAuthSets} we obtain 
\ben
S(1,2)+S(1,3)&\geq& S(1,2,3)+S(1)+2S(S).
\een
By the subadditivity inequality we have 
\ben
S(1)+S(3)&\geq&S(1,3).
\een
Adding the previous two equations we obtain
\ben
S(1,2)+S(3)&\geq& S(1,2,3)+2S(S),
\een
which can be rewritten as 
\ben
S(1,2)&\geq&S(1,2,3)-S(3)+2S(S) \label{eq:a1a2}.
\een
Now let us add the equations~\eqref{eq:a0a1},~and~\eqref{eq:a1a2} 
to obtain
\ben
S(1,2)+S(3,4)&\geq &S(1,2,3,4)+4S(S). 
\een
By Lemma~\ref{lm:jointEntropy}, $S(1,2,3,4)=S(p)+S(S)\geq 2S(S)$.
Further, applying the subadditivity inequality to $S(1,2) $ and $S(3,4)$ 
we can reduce the previous equation to 
\ben
S(1)+S(2)+S(3)+S(4) &\geq&6S(S)
\een

Finally,  the non-negativity of $S(\cdot)$ implies at least one of the entropies $S(i)$, $1\leq i\leq 4$
must be greater than their average $6S(S)/4$. Thus at least some share of $\Gamma^{(4)}$ 
and $\overline{\Gamma}^{(4)}$ must at least 
$50\%$ larger than the size of the secret. 
\end{proof}

Theorem~\ref{th:shareSize} predicts that $\overline{\Gamma}^{(4)}$ will have an information rate of $5/7\approx 0.714$, while 
the above result shows that it can be at most $2/3\approx 0.667$.

\section{Conclusion}
In this paper we investigated the sizes of shares for quantum secret sharing schemes. 
We showed that there exist access structures, on $n$ participants, for which the size of the share grows  as $O(n/\log_2 n)$. 
To the best of our knowledge, these bounds represent the strongest lower bounds on the
size of a share (equivalently upper bounds on the information rate). Some questions suggested by
these results are the tightness of these bounds and the schemes for realizing these access
structures with information rates close to the bounds. We will address these questions elsewhere. We hope that this work highlights the fact 
that there is a significant gap between the upper bounds and lower bounds for the information
rate of (perfect) quantum secret sharing schemes and motivates further  research in this direction.  

\section*{Acknowledgment}
This research is sponsored by CIFAR, MITACS and NSERC. I would like to thank 
Robert Raussendorf for his support and encouragement. 


\end{document}